\documentclass[journal]{IEEEtran}
\usepackage{cite}

\usepackage{color,soul}
\usepackage{braket}
\usepackage{array}
\usepackage{adjustbox}
\usepackage{subcaption}
\usepackage{multirow}
\usepackage{booktabs}
\usepackage{amssymb,amsthm}
\usepackage{graphicx}

\newtheorem{theorem}{Theorem}

\newtheorem{corollary}{Corollary}

\renewenvironment{IEEEbiography}[1]
{\IEEEbiographynophoto{#1}}
{\endIEEEbiographynophoto}

\ifCLASSINFOpdf
\else
\fi
\usepackage{amsmath}
\usepackage{mathtools}
\usepackage{amssymb}
\usepackage{dsfont}
\usepackage[font=footnotesize]{caption}
\usepackage{algorithmic}
	\usepackage{url}
	
	\usepackage{mathtools}
	\usepackage{makecell}

\begin{document}
		%
		\title{Information Density as a Quantitative Measure for AI-enabled Virtual Sensing: Feasibility and Limits}
		
		%
		%
		%
		
		\author{Hrishikesh Dutta,~\IEEEmembership{Member,~IEEE,}
			Roberto Minerva,~\IEEEmembership{Senior Member,~IEEE,}
			Reza Farahbakhsh,~\IEEEmembership{Member,~IEEE,}
			and~Noel Crespi,~\IEEEmembership{Senior Member,~IEEE}
			\thanks{The authors are with the Data Intelligence and Communication Engineering Lab, Telecom SudParis, Institut Polytechnique de Paris, France (email: hrishikesh.dutta@telecom-sudparis.eu; roberto.minerva@telecom-sudparis.eu; reza.farahbakhsh@it-sudparis.eu, noel.crespi@telecom-sudparis.eu)}
		}
		
		%
		%

	\maketitle
	
	\begin{abstract}
		Modern IoT and sensor networks generate vast amounts of data, posing significant challenges for storage, transmission, and real-time processing. Traditional approaches, such as compressive sensing and machine learning-based compression, often suffer from computational inefficiencies and irreversible data loss. This paper introduces Information Density as a quantitative metric to support sensor deployment and enable AI-driven virtual sensing. We propose a framework that leverages spatial, temporal and inter-modal correlations among sensor signals to perform sensing tasks even in the absence of physical sensors. Two complementary measures: (i) Phase in Eigen Space and (ii) Mutual Information, are developed to quantify and assess information density, enabling the selection of optimal sensor configurations across both intra-modality and cross-modality scenarios. Validated using real-world data from Madrid's smart city infrastructure, this framework demonstrates the feasibility of replacing physical sensors with virtual ones under bounded error conditions (e.g., achieving $<3.21\%$ mean error with a single sensor). The results highlight the potential for scalable and energy-efficient sensing systems in smart environments.
	\end{abstract}
	

	%
	\IEEEpeerreviewmaketitle

	\section{Introduction}
	\label{intro}
	
	Modern day technical advancements, such as, metaverse and smart environments, are fueled by massive data generation from large scale networks of sensors and IoT devices. Such massive IoT infrastructures, powering systems like autonomous transportation, smart cities, and industrial automation, are generating data at rates that exceed the capacity of conventional data infrastructures. As the number of connected devices grows, the volume of data, the speed at which it is generated, and the diversity of data types pose serious challenges for storage, transmission, and real-time processing. Without scalable management strategies, these systems face increased latency, higher energy consumption, and degraded performance. Let us consider Madrid, for instance, a front-runner in the smart city landscape. The city has deployed a vast and diverse array of IoT devices across its urban environment, enabling multi-modal sensing in areas such as traffic, air quality, weather conditions, noise levels, energy usage, and waste management. Its traffic monitoring system alone integrates over 7,000 sensors distributed across more than 4,000 locations, generating more than 145 million data points \cite{ding2024deep}. These figures represent just one slice of the city’s overall sensor infrastructure. With numerous other sensing networks operating simultaneously, the volume and complexity of the collected data become significant. In such a densely instrumented environment, the need for scalable data management and efficient system coordination becomes critical.
	
	There have been several research efforts to efficiently manage large-scale sensor networks and the massive data they generate. Traditional methods like compressive sensing \cite{abbasian2020survey, ahmad2019lossless, arumugam2015ee, cao2020multi, chen2020new, chen2019hierarchical, chen2019layered, gao2022application, halder2019limca, jarwan2019data, ketshabetswe2021data, roy2022limited, saidani2019new} aggregate data from nearby sensors using spatio-temporal correlations and lossless compression, but often struggle with adaptability and impose heavy computation on sensor nodes. To address this, AI-based approaches have been adopted, including dimensionality reduction techniques like PCA \cite{chowdhury2020adaptive, diwakaran2019cluster, yang2023guidelines}, SVD \cite{alam2021error, he2019multi}, and LDA \cite{fabiyi2021folded}, along with models such as autoencoders \cite{chen2020wsn, kuester2023convolutional} and reinforcement learning \cite{yun2021q}. More recent developments explore semantic communication \cite{feng2023data, wang2023semantic} and synthetic or general-purpose sensing \cite{laput2017synthetic, laput2019sensing, zhang2018vibrosight}, which use a limited number of versatile sensors to infer multiple phenomena. However, many of these methods place significant processing demands on resource-constrained devices. As a solution, the Data-driven Modality Fusion (DMF) framework \cite{dutta2025data} uses machine learning to fuse data from multiple sensors, enabling virtual sensing and reducing physical sensor count. Additional works \cite{ding2024deep, farid2023data, costa2023sensor, song2023multi, li2023submodularity} focus on minimizing sensor deployments while maintaining acceptable sensing performance.

	All these frameworks - relying on lossless compression, correlation measures and machine learning techniques - serve as noteworthy milestones towards an efficient management of large scale sensor network infrastructures and handling the massive volume of data generated. Particularly, the approaches developed in \cite{dutta2025data} empirically demonstrate their capabilities to retain meaningful feature information for data reconstruction and inference, without the complete reliance on physical sensors. Nevertheless, one common problem with these developments is the lack of a quantitative measure or a criterion that can be used to explore the optimal number and distribution of physical sensors in one or more regions, for a given sensing performance. This is not as trivial as it appears in the first look, since, the elimination of a subset of certain highly correlated sensing modalities can be useful from data management perspective, while still allowing us to recover them from the other modalities (as shown in our prior work \cite{dutta2025data, dutta2025cross}); on the other hand, these highly correlated sensor values can be useful for extracting spatial information  (as shown in \cite{ding2024deep}) for a network deployed in a large city like Madrid. Therefore, in this work, we first extend the concept of Data-driven Modality Fusion (DMF) proposed in \cite{dutta2025data} to incorporate spatial sensing correlation into account and demonstrate how DMF allows us to perform sensing in a region with virtual, AI-enabled sensors, in the absence of real, physical sensors. Building on that, this paper introduces Information Density as a novel, quantitative metric to optimize sensor configurations. This metric enables informed decisions about which physical sensors are essential and which can be virtually inferred, thus reducing cost, energy consumption, and data redundancy. The added value lies in its application to both intra-modality (e.g., temperature-to-temperature inference) and cross-modality (e.g., air quality inferred from noise or traffic data) virtual sensing. Through the lens of Madrid’s sensor infrastructure, which is the use-case considered for this work, we demonstrate that accurate sensing can be achieved with significantly fewer physical devices. 
	

	This paper has the following specific scopes and contributions.
	\begin{itemize}
		\item This paper extends the concept of DMF to incorporate spatial sensing correlation, enabling fusion-based sensing in the absence of physical sensors. We demonstrate how virtual sensors, trained using fused data from multiple physical modalities, can replace or supplement real sensors for large-scale sensing applications.
		\item This paper introduces the concept of information density as a quantitative metric to estimate the optimal number, type, and spatial distribution of sensors required for effective virtual sensing. It identifies and analyzes the distinct approaches needed for virtual sensing across sensors of the same modality (\textit{intra-modality virtual sensing}), and of different modalities (\textit{cross-modality inference}), based on their spatial and semantic correlations.
		\item The proposed claims are validated in a real-world context by referencing smart city environments of Madrid, highlighting the relevance and scalability of the proposed methods.
	\end{itemize}
	
	The remainder of the paper is organized as follows. First, we present a literature review of the existing approaches of data and network management in the context of large scale IoT/sensor networks. Next, we formally present the problem of Information Density and its relevance in the context of network and data management problem, in section \ref{prob_def}, followed by the proposed quantitative measures for defining it in section \ref{measures}. Then, we present a use case of smart city IoT network deployed in Madrid in section \ref{use_case} to demonstrate the proposed concept of virtual sensing and information density, where we present a detailed performance evaluation of the framework and the models. Finally, the conclusions of the study are reported in section \ref{conc}.

	
	\section{Related Work}
	\label{secii}
	
	There have been several research works with focus on improving data management in sensor networks, with an emphasis on reducing communication overhead and extending the operational lifespan of energy-constrained nodes. Broadly, two major strategies have emerged: local compression and distributed compression. Distributed schemes consolidate data from multiple sensor nodes to eliminate redundancy before transmission, typically relying on spatial correlations \cite{abbasian2020survey, arumugam2015ee, halder2019limca, chen2019layered}. However, this assumption does not hold in all deployment scenarios, limiting the applicability of such methods. In contrast, local compression techniques focus on temporal correlations within individual sensors to minimize transmitted data \cite{cao2020multi, jarwan2019data}, with hybrid strategies combining both temporal and spatial cues also explored \cite{chen2019hierarchical}.
	
	Compression techniques can be further classified into lossless and lossy methods. Lossless schemes, such as those based on entropy encoding and statistical modeling \cite{saidani2019new, gao2022application, ahmad2019lossless}, ensure data fidelity but may not yield significant reduction ratios. Lossy techniques, including those based on Bayesian prediction and segmentation algorithms \cite{chen2020new, roy2022limited}, achieve greater compression at the cost of irreversible data degradation and increased computational demand.
	
	To address limitations of traditional compression, recent studies have explored machine learning-based approaches for data reduction in IoT and sensor networks. Feature extraction and dimensionality reduction techniques, such as Principal Component Analysis (PCA) \cite{chowdhury2020adaptive, yang2023guidelines, diwakaran2019cluster}, Singular Value Decomposition (SVD) \cite{alam2021error, he2019multi}, and Linear Discriminant Analysis (LDA) \cite{fabiyi2021folded}, have been adopted to project high-dimensional sensor data into a compact latent space. Despite their effectiveness, these methods require access to the full dataset, which is often infeasible in memory-limited sensor nodes. Furthermore, neural network-based models like autoencoders \cite{kuester2023convolutional, chen2020wsn} and reinforcement learning strategies \cite{yun2021q} offer adaptive alternatives but introduce higher computational requirements that many embedded devices cannot support. Semantic communication models \cite{feng2023data, wang2023semantic} have also emerged, focusing on optimizing transmission by encoding data in terms of its relevance to downstream tasks rather than raw accuracy.
	
	Another notable direction involves the paradigm of synthetic or general-purpose sensing, where a small number of high-capability sensors capture ambient information used to infer multiple phenomena \cite{laput2017synthetic, laput2019sensing, zhang2018vibrosight}. These systems aim to reduce hardware redundancy by learning indirect mappings between sensor signals and real-world events, thereby lowering deployment complexity and potentially improving user privacy. However, these systems often rely on intensive on-device computation and may fall short in dynamic environments such as urban-scale IoT deployments.
	
	More recently, emphasis has shifted from data compression alone to the broader goal of optimizing sensor infrastructure itself. Techniques such as Data-driven Modality Fusion (DMF) \cite{dutta2025data} aim to reduce physical sensor deployments by combining information from multiple modalities using machine learning, creating virtual sensors capable of accurate inference from fused data streams. This approach addresses scalability and energy-efficiency by offloading computation and reducing redundant sensor usage. In parallel, other works \cite{ding2024deep, costa2023sensor, song2023multi, li2023submodularity, farid2023data} explore optimal sensor placement strategies to minimize coverage gaps while preserving sensing accuracy. These methods demonstrate promising results in urban-scale scenarios by maintaining performance with fewer deployed sensors.
	
	There are recent related approaches for virtual sensing as well. The authors in \cite{zhao2025graph} introduce a novel heterogeneous temporal Graph Neural Network for virtual sensing that performs sensor fusion across diverse modalities in complex systems. A virtual sensing-enabled digital twin framework is proposed in \cite{hossain2025virtual} for real-time monitoring of nuclear systems leveraging deep neural operators. Similarly an LSTM-based virtual sensing technique for industrial process control loop is designed in \cite{gonzalez2025assessment}. The paper \cite{hu2025online} proposes a virtual sensing technique based on kriging interpolation, which gives an unbiased sound pressure estimator with minimum variance, assuming the spatial correlation model. The proposed method only depends on the distances between physical and virtual microphones and makes no assumption on the primary noise field, implying that the method naturally tracks the variations of the noise sources.
	
	While these developments mark significant progress, they often overlook systematic frameworks for determining the optimal number and type of sensors needed in a given region. The lack of formal metrics or optimization strategies to quantify sensing adequacy, especially when substituting physical sensors with virtual ones, presents a notable gap. This paper builds upon the DMF framework to introduce the concept of information density, a formalized criterion for assessing the sufficiency and spatial distribution of sensing resources. This approach offers a pathway to adaptive, AI-enhanced sensing infrastructures that can scale effectively with smart city requirements.


	\section{Information Density: Problem Definition}
	\label{prob_def}
	
	Let \( \mathcal{E} \) denote a physical environment represented by the spatial coordinates \( (r, \theta) \in \mathbb{R}^2 \). Within \( \mathcal{E} \), a set of physical sensors \( \mathcal{S} = \{s_{(r, \theta)}\} \) are deployed, each associated with a sensing modality \( m_i \in \mathcal{M} \). These sensors generate time-indexed data streams \( S^j_{(r,\theta)}(t) \), where \( j = 1, \ldots, M_{(r,\theta)} , \ldots |\mathcal{M}|\) denotes the number of sensing modalities available at location \( (r,\theta) \). In addition, external data sources \( \mathcal{D} \) (e.g., social data, inferred metadata, cloud services) may provide auxiliary context about phenomena in the environment.
	
	
	We aim to infer or describe a set of real-world phenomena \( \Phi = \{\phi_1, \phi_2, \ldots\} \) occurring within \( \mathcal{E} \), using sensor and auxiliary data. For a specific phenomenon \( \phi \in \Phi \), we define an \textit{information inference function}:
	
	\begin{equation}
		I_k(\phi) = f_k(\mathcal{S}^\prime, \mathcal{D}^\prime), \quad \mathcal{S}^\prime \subseteq \mathcal{S}, \quad \mathcal{D}^\prime \subseteq \mathcal{D}, \quad 1 \leq k \leq P
		\label{eq:inf_func}
	\end{equation}
	
	Each function \( f_k \) attempts to infer or model some aspect of \( \phi \) using selected subsets of sensor and auxiliary data. An environment is considered \textit{information-dense} for phenomenon \( \phi \) if such a set of functions \( \{I_k\} \) can accurately describe it, using minimal number of sensors and data.
	
	To quantify \textit{information density}, we first pose the following optimization problem. Let \( \mathcal{S}_{\text{opt}} \subseteq \mathcal{S} \times  \mathcal{M}\) denote the minimal subset of sensors that, through fusion, allow accurate inference of all \( \phi \in \Phi \). Our goal is to minimize the total number of deployed modalities:
	
	\begin{equation}
		{\mathcal{S}_{\text{opt}}} = \min \sum_{(r, \theta)} (s_{(r, \theta)}\times M_{(r, \theta)})
		\label{eq:minimize_modalities}
	\end{equation}
	
	Subject to the condition:
	
	\begin{equation}
		\Delta \left( I_k^{(i)} = f_k(\mathcal{S}_i^\prime, \mathcal{D}^\prime),\; I_k^{(j)} = f_k(\mathcal{S}_j^\prime, \mathcal{D}^\prime) \right) < \epsilon,\quad \forall i \neq j
		\label{eq:divergence_constraint}
	\end{equation}
	
	Where:
	\begin{itemize}
		\item \( \mathcal{S}_i^\prime \subseteq \mathcal{S} \) is the set of sensors available at location \( (r_i, \theta_i) \),
		\item \( \Delta(\cdot,\cdot) \) is a divergence or distance measure between inference outputs at different spatial locations,
		\item \( \epsilon \) is the maximum tolerable divergence threshold for acceptable performance.
	\end{itemize}
	
	This formulation ensures that different locations across \( \mathcal{E} \) yield consistent, accurate models of \( \phi \), despite relying on reduced sensor configurations. Note that this formulation does not assert global optimality, but instead defines a consistency criterion under which reduced sensor configurations can be considered informationally sufficient for a given phenomenon.
	
	By solving this problem, we identify:
	\begin{itemize}
		\item The minimum number and optimal spatial distribution of sensors needed,
		\item The conditions under which physical sensors can be substituted by virtual or fused sensors without significant loss of information.
	\end{itemize}
	
	This information density framework serves as a foundation for designing scalable, cost-efficient, AI-driven sensing infrastructures. It is applicable both to \textit{intra-modality virtual sensing} (e.g., reconstructing temperature from fewer temperature sensors) and \textit{cross-modality inference} (e.g., estimating air quality using vibration and sound data), each requiring tailored fusion strategies.

	\begin{figure*}[t]
		\centering
		\includegraphics[width=0.8\textwidth]{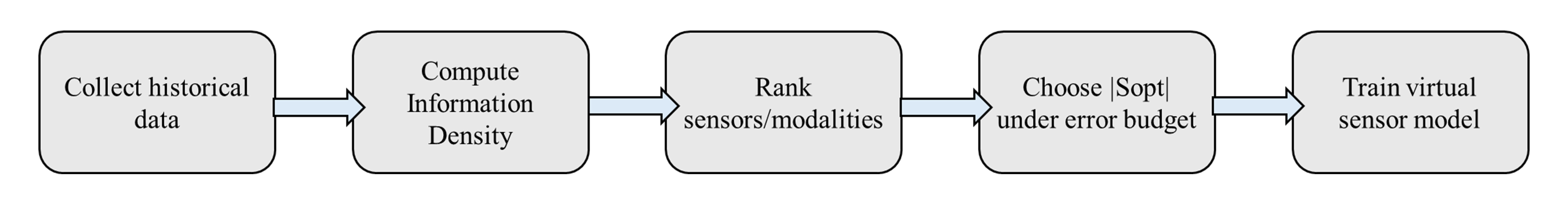}
		\caption{Practitioner workflow for operating Information Density}
		\label{fig:workflow}
	\end{figure*}

	\section{Measures for Defining Information Density}
	\label{measures}
	
	In this section, we propose two measures for defining Information Density. Before providing the formal definition, we introduce the two key concepts of sensing virtualization, viz., \textit{intra-modality virtual sensing} and \textit{cross-modality inference}, and point-out how the definition of Information Density needs to be adapted differently for these two contexts.
	
	\textbf{\textit{Intra-modality virtual sensing (ImVS)}} refers to the scenario of sensing inference executed on the sensors that sample data from the same modality. More precisely, when there are multiple sensors measuring a phenomenon (represented by a specific sensing modality) in a given environment, there exists an optimal set of sensing placements that can be best used to derive and model the phenomenon, with a tolerable error margin. This can be achieved using an extended version of Data-driven Modality Fusion \cite{dutta2025data}, as we present this later in section \ref{use_case}. An example of this scenario would be to find out a spatio-temporal temperature profiling of an environment by positioning temperature sensors at different geo-locations. In order to find the optimal set of sensors for this case, it is desirable to find out the ones that have high divergence measures in their sensed data. In other words, the optimal set of sensors should be selected such that the measured values or data possess low similarity in the feature space or the phenomenon that they are sensing at a given time instant. The idea here would be to estimate or infer the sensing measurements at all other instances, based on the high similarity measures that exist with readings from other sensors. This is motivated by the fact if the data  follows similar distribution or there exists strong relationship that can be modeled or function-approximated, then the placements of physical sensors at multiple locations can be got rid of, thus allowing redundancy reduction. 
	
	\textbf{\textit{Cross-modality Inference (CmI)}}, on the other hand, is applicable to scenarios where the modeling of an environment requires information about several phenomena. Such information is captured by sensing across multiple different modalities, usually at a specific geo-location. An example of such scenario can be referred to smart agriculture for monitoring different meteorological information, such as, temperature, humidity, radiation, along with soil moisture levels and other conditions related to plant growth, using multi-modal sensing capabilities. The idea here is that if there exists high similarity (or low divergence) in the data distribution along the features of interest, then one modality can be estimated from the other, as accomplished successfully in \cite{dutta2025data}, drastically removing the need for multitude of sensors along different modalities.
	
	In order to find out the best set of sensors and/or sensing modalities that can effectively represent and model the phenomenon of interest, we need to rely on some statistical similarity or divergence measures. In the following part of this section, we define two such measures that can quantify information density and allow us to find the optimal set of sensors or sensing modalities.
	
	\textbf{\textit{Phase in Eigen Space:}} For a given dataset, eigen vectors represent the direction along which it possesses the maximum variance. Hence, the idea here is that if an Eigen space representation of the data collected across different sensors or sensing modalities is created, then the phase angles among their largest Eigen vectors would represent the divergence in the data. In other words, larger the phase difference between the vectors (more orthogonal), less likely is that the sensor measurements have similarity and consequently lower is the information density. Formally, this can be computed as follows.
	
	For a given sensor $s_m$, we compute the Eigen vector $\mathbf{x}_m^*$ with the largest Eigen value $\lambda_m^*=max_p{\lambda_m(p)}$. This Eigen vector corresponding the largest Eigen value $\lambda_m^*$ can be derived using Eqn. \ref{eigen_eqn} 
	
	\begin{equation}
		\mathbf{x}_m^*=argmax_{x_m}{(\mathbf{x}_m^TC_m\mathbf{x}_m)}
		\label{eigen_eqn}
	\end{equation}
	
	Here $C_m$ is the covariance matrix of the normalized dataset ($\mathcal{D}_m$), related to the $m^{th}$ sensor, with $n_{\mathcal{D}_m}$ samples, with eigenvalues $\lambda_1 \geq \lambda_2 \geq \cdots \geq \lambda_d$ and corresponding eigenvectors $v_1,\ldots,v_d$, which can be computed as follows.

	\begin{equation}
		C_m = \frac{1}{n_{\mathcal{D}_m}-1}(\mathcal{D}_m-\mathbb{E}[\mathcal{D}_m])^T(\mathcal{D}_m-\mathbb{E}[\mathcal{D}_m])^T
		\label{covariance_eqn}
	\end{equation}

	Next, considering the entire set of sensors $ \tilde{\mathcal{S}} = \{ s_1, s_2, \ldots, s_N\} $, where $s_k \in \mathcal{S} \times  \mathcal{M}$, we compute the angles between the principal Eigen vectors. This gives us the set $\Omega = \{\omega_{(1,2)}, \omega_{(1,3)}, \ldots, \omega_{(N,N-1)}\}$, where $\omega_{(m,n)}$ represent the angle between sensors $m, n$ computed as:
	
	\begin{equation}
		\omega_{(m,n)}={cos^{-1}(\frac{\mathbf{x}_m^*.\mathbf{x}_n^*}{||\mathbf{x}_m^*||.||\mathbf{x}_n^*||})}
		\label{angle_comp}
	\end{equation}
	
	This space $\Omega$ comprising of the phase among the principal Eigen vectors in the sensor readings, physically represent the 2D scalar field of the Information Density for a given environment \( \mathcal{E} \) and a specific phenomenon \( \Phi \). The more orthogonal the vectors are (that is, higher the values of \(\omega_{(m,n)}\)), the less is the similarity among the collected data, and less is the information density.
	
	Now, one way to find the optimal set of sensing modalities could be based on ranks assigned to them depending on the average angle they form with all other sensors. In that case, the sensors with the lowest average angles are deemed to exhibit the most representative or directional patterns. Determining the set $\mathcal{S}_{\text{opt}}^k$ with top $k$ sensors is given by Eqn. \ref{phase_min_eqn}.

	\begin{equation}
		\mathcal{S}_{\text{opt}}^k= argsort(\frac{1}{|\tilde{\mathcal{S}}|-1} \sum_{n\in \tilde{\mathcal{S}}, n \neq m}{\omega_{(m,n)}})_{[:k]}
		%
		\label{phase_min_eqn}
	\end{equation}


	\begin{figure*}[t]
		\centering
		\begin{subfigure}[t]{0.48\textwidth}
			\centering
			\includegraphics[width=\linewidth]{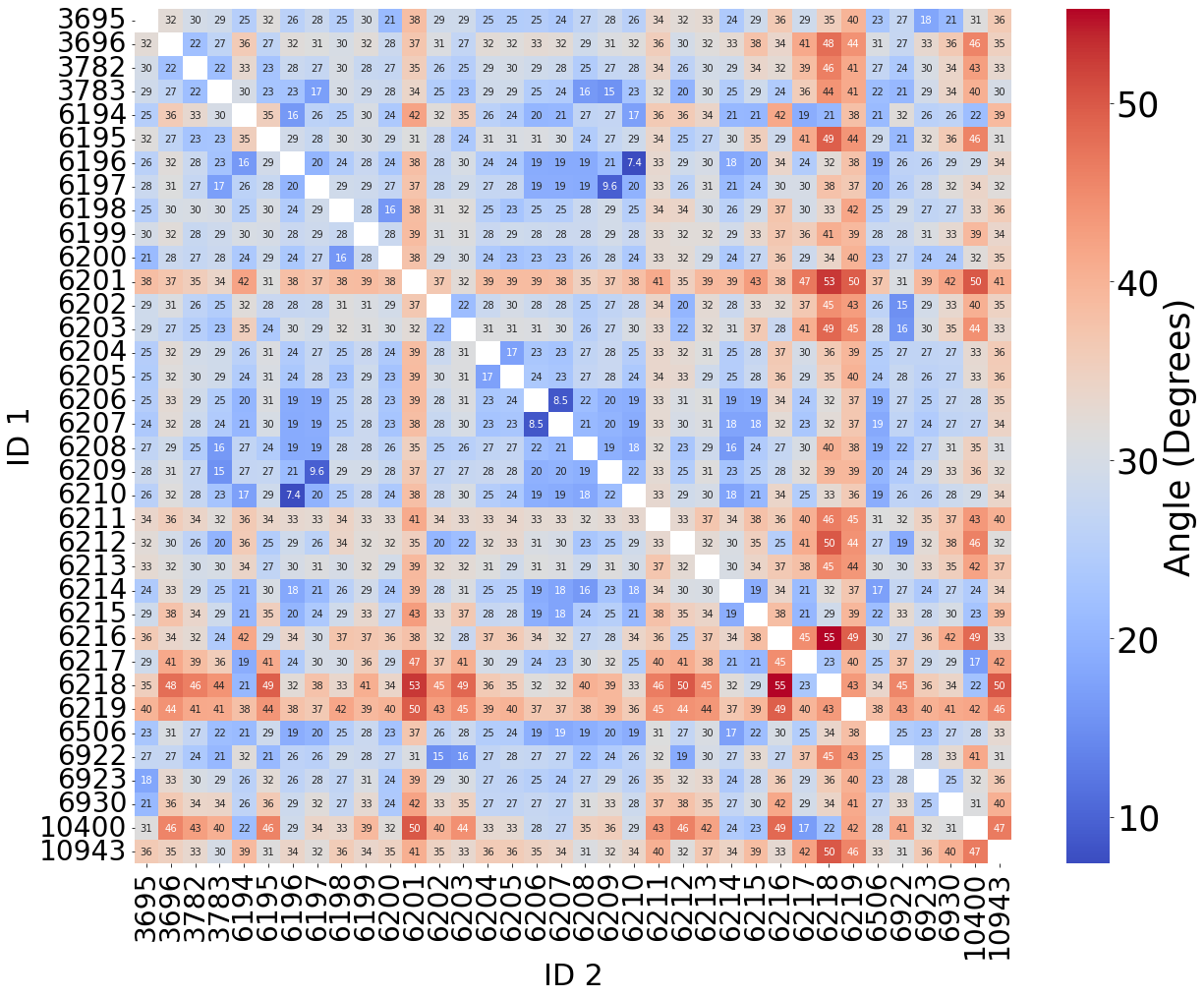}
			\caption{Eigen Space phase}
			\label{fig:phase_2d}
		\end{subfigure}
		\hfill
		\begin{subfigure}[t]{0.48\textwidth}
			\centering
			\includegraphics[width=\linewidth]{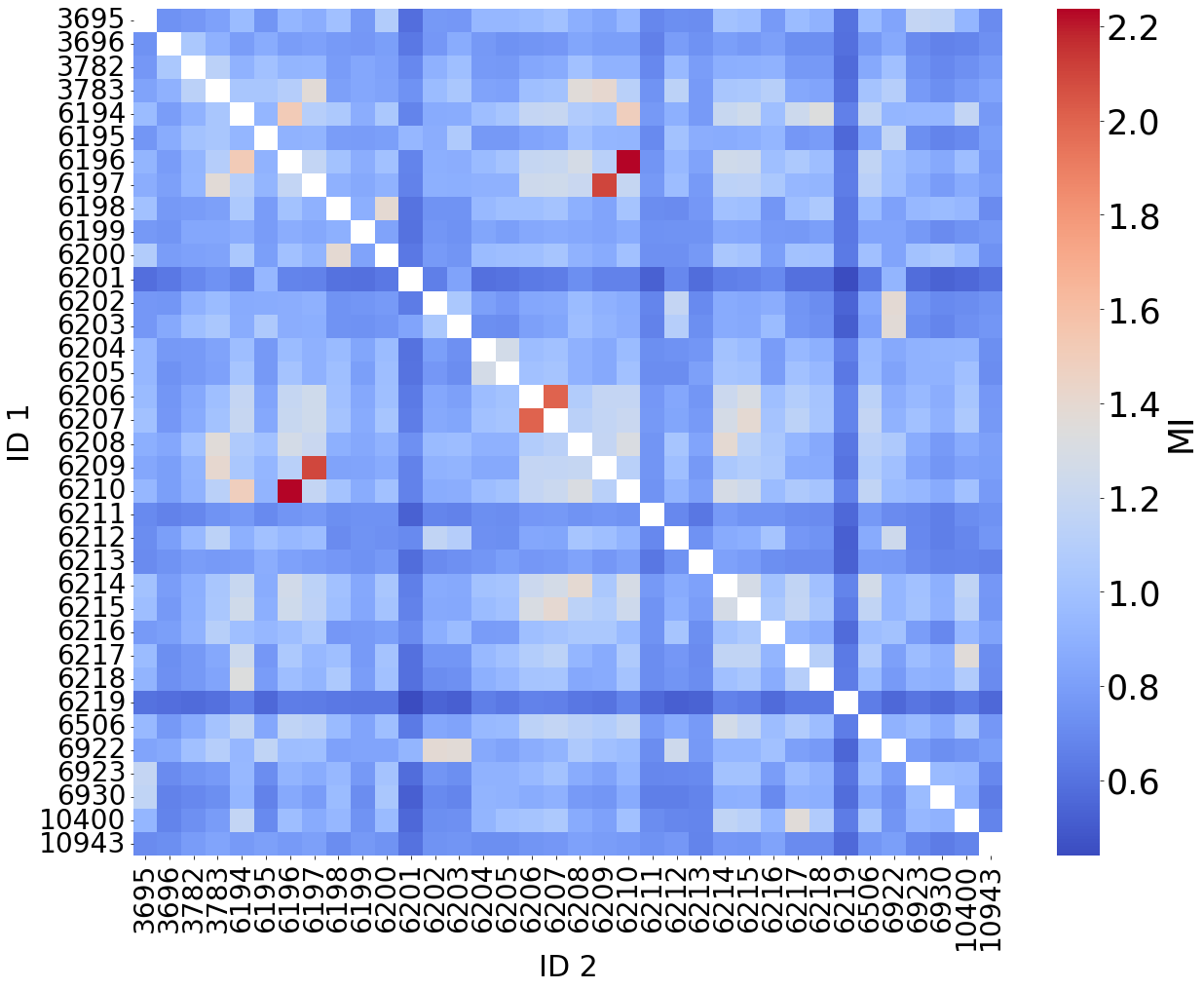}
			\caption{Mutual Information}
			\label{fig:mi_2d}
		\end{subfigure}
		\caption{(a) Spatial distribution of information density in Madrid’s District 19 traffic network using Phase in Eigen Space. Colors indicate the angular divergence between principal components of sensor data; smaller angles (blue) suggest high redundancy, while larger angles (yellow/red) indicate diverse, complementary information. (b) Spatial distribution of information density using Mutual Information. Higher mutual information values (red) reflect strong statistical similarity between sensor readings, implying potential redundancy. This complements Figure 1a by offering an inverse yet consistent view.}
		\label{fig:comparison_2d}
	\end{figure*}

	Now we are in a position to state the following theorem.
	
	\begin{theorem}
		The principal eigenvector angle $\omega_{(m,n)}$ between sensors $m$ and $n$ quantifies the alignment of the dominant variance directions of the two sensor signals.
	\end{theorem}

	\textbf{Assumptions:}
	
	The interpretation of the principal eigenvector angle relies on the following assumptions:
	\begin{enumerate}
		\item The sensor signals are weakly stationary over the observation window, so that empirical covariance estimates are meaningful.
		\item Sensor data are represented in a common feature space with compatible scaling and normalization.
		\item The dominant eigenvalue of each covariance matrix is well separated from the remaining spectrum, ensuring that the principal eigenvector captures a stable variance direction.
		\item The covariance structure is representative of the sensing process over the time horizon of interest.
	\end{enumerate}

	\begin{proof}
		
		The principal eigenvector $v_1^m$ maximizes the Rayleigh quotient:
		\begin{equation}
			v_1^m = \operatorname*{argmax}_{\|v\|=1} v^T C_m v
		\end{equation}
		
		The angle between principal components is:
		\begin{equation}
			\cos(\omega_{(m,n)}) = \frac{v_1^m \cdot v_1^n}{\|v_1^m\|\|v_1^n\|}
		\end{equation}
		
		By the Cauchy-Schwarz inequality:
		\begin{equation}
			0 \leq |\cos(\omega_{(m,n)})| \leq 1
		\end{equation}
		
		For $\omega_{(m,n)} \approx 0^\circ$:
		\begin{equation}
			v_1^m \approx v_1^n \Rightarrow \text{high redundancy}
		\end{equation}
		
		For $\omega_{(m,n)} \approx 90^\circ$:
		\begin{equation}
			v_1^m \perp v_1^n \Rightarrow \text{complementary information}
		\end{equation}
		
		The optimal set $\mathcal{S}_{opt}^k$ minimizes:
		\begin{equation}
			\frac{1}{|\tilde{\mathcal{S}}|-1}\sum\omega_{(m,n)}
		\end{equation}
		which maximizes linear independence in the selected set.

	\end{proof}
	
	\begin{corollary}
		The average principal angle is inversely proportional to information density:
		\begin{equation}
			ID_{angle} \propto \frac{1}{\text{mean}(\omega_{(m,n)})}
		\end{equation}
	\end{corollary}

	\textbf{\textit{Mutual Information:}} Information Density can also be defined by considering Shannon's Entropy for measuring the information between two random variables \cite{lin1991divergence}. This can be obtained by considering two random variables $\mathbf{x_m}$ and $\mathbf{x_n}$ and we associate them with two sensors $m, n$ from the set $\tilde{\mathcal{S}}$. Next, we compute the mutual information between $\mathbf{x_m}$ and $\mathbf{x_n}$ given by Eqn. \ref{mi__eqn}, where $ p(\mathbf{x_m}, \mathbf{x_n})$ refers to the joint probability distribution of $\mathbf{x_m}$ and $\mathbf{x_n}$; and $p(\mathbf{x_m})\,p(\mathbf{x_n})$ represent the respective marginal probability distributions.

	\begin{equation}
		I(\mathbf{x}_m; \mathbf{x}_n) = \sum_{\mathbf{x}_m} \sum_{\mathbf{x}_n} p(\mathbf{x}_m, \mathbf{x}_n) \log \left( \frac{p(\mathbf{x}_m, \mathbf{x}_n)}{p(\mathbf{x}_m)\,p(\mathbf{x}_n)} \right)
		\label{mi__eqn}
	\end{equation}

	Next, we compute the mutual information for all the sensors in the set $\tilde{\mathcal{S}}$. This gives us the set $\Gamma = \{\gamma_{(1,2)}, \gamma_{(1,3)}, \ldots, \gamma_{(N,N-1)}\}$, where $\gamma_{(m,n)}$ represent the mutual information between sensors $m, n$ computed using Eqn. \ref{mi__eqn}. That is, 
	
	\begin{equation}
		\gamma_{(m,n)} \equiv I(\mathbf{x}_m; \mathbf{x}_n)
	\end{equation}

	Similar to the space $\Omega$, this mutual information space $\Gamma$ also physically represent the 2D scalar field of the Information Density for a given environment \( \mathcal{E} \) and a specific phenomenon \( \Phi \), but inversely. The higher the values of \(\gamma_{(m,n)}\), the more is the similarity among the sensor readings, and higher is the information density.
	
	Thus, in order to find the optimal set of sensing modalities using the space  $\Gamma$, the sensor ranks are assigned based on the mutual information values. The sensors with the highest mutual information are deemed to exhibit the most representative or directional patterns. Determining the set $\mathcal{S}_{\text{opt}}^k$ with top $k$ sensors is given by Eqn. \ref{mi_max_eqn}.

	\begin{equation}
		\mathcal{S}_{\text{opt}}^k= argsort(\frac{1}{|\tilde{\mathcal{S}}|-1} \sum_{n\in \tilde{\mathcal{S}}, n \neq m}{\gamma_{(m,n)}})_{[k:]}
		%
		\label{mi_max_eqn}
	\end{equation}
	
	The following theorem can be stated with respect to this measure.

	\begin{theorem}
		The mutual information $I(X_m;X_n)$ measures the non-linear statistical dependence between sensors.
	\end{theorem}
	
	\begin{proof}
		
		Forward Direction: \\
		If $I(X_m; X_n) = 0$, then by the properties of Kullback-Leibler divergence:
		\begin{equation*}
			D_{\text{KL}}(p(x_m, x_n) \parallel p(x_m)p(x_n)) = 0
		\end{equation*}
		\begin{equation}
			\implies p(x_m, x_n) = p(x_m)p(x_n) 
		\end{equation}
		
		Thus, $X_m \perp X_n$.

		Reverse Direction: \\
		
		If $X_m \perp X_n$, then:
		\begin{equation}
			I(X_m; X_n) = \mathbb{E} \left[ \log \frac{p(x_m)p(x_n)}{p(x_m)p(x_n)} \right] = \mathbb{E}[\log 1] = 0
		\end{equation}
		
	\end{proof}
	
	\begin{corollary}
		The average mutual information is directly proportional to information density:
		\begin{equation}
			ID_{MI} \propto \text{mean}(I(X_m;X_n))
		\end{equation}
	\end{corollary}

	\textbf{\textit{Comparison of Information Density Measures}}: Although both principal eigenvector alignment and mutual information quantify aspects of Information Density, they capture complementary properties of sensor data and are not interchangeable (Table \ref{id_comp}).
	
	Principal eigenvector alignment primarily reflects linear structure and dominant variance directions in sensor signals. It is computationally efficient, robust under moderate noise, and well suited for large-scale deployments where frequent recomputation is required. However, it may fail to capture non-linear dependencies or higher-order interactions between sensors.
	
	Mutual information, by contrast, captures both linear and non-linear statistical dependence and provides a model-agnostic measure of shared information. This makes it particularly valuable for cross-modality inference. Its limitations include higher computational cost and sensitivity to sample size, which can lead to unreliable estimates in data-scarce regimes.
	
	In practice, eigen-space measures are preferable for fast, scalable sensor ranking in dense networks, while mutual information is more appropriate when cross-modal semantic relationships are of interest and sufficient data is available. The experimental results in Section \ref{use_case} demonstrate that, in data-rich urban sensing environments, both measures produce consistent rankings, thereby reinforcing the robustness of the proposed framework.

	\begin{table}[]\centering
		\caption{Comparison of Information Density Measures}
		\label{id_comp}
		\renewcommand{\arraystretch}{2}
		\begin{tabular}{|c|c|c|}
			\hline
			\textbf{Aspect}      & \textbf{Eigen-space Phase}       & \textbf{Mutual Information}\\ \hline 
			Dependency Type           & Linear           & Linear + Non Linear            \\ \hline 
			Data Requirement 	& Moderate 	& High           \\ \hline 
			Computational Cost           & Low    	& High      \\ \hline 
			Failure Mode              & Non-linear Coupling		& Small sample bias             \\ \hline 
		\end{tabular}
	\end{table}

	%

	%
	%
	%
	%
	%

	\section{Operational Meaning and Decision Role of Information Density}
	\label{op_meaning}
	
	Information Density is introduced in this work as a decision-support metric rather than a decision-theoretic optimum. Its purpose is not to directly solve an optimal sensor placement or replacement problem under all conditions, but to provide a quantitative and computationally tractable proxy for identifying redundancy, substitutability, and representational sufficiency among sensing resources.
	
	Concretely, Information Density characterizes how much of the dominant structure of a sensing field, captured through variance alignment or statistical dependence, can be preserved when a subset of physical sensors or modalities is retained. High Information Density between sensors or modalities indicates that one can be inferred from others with bounded error, enabling informed decisions about sensor removal, replacement by virtual sensors, or consolidation of sensing modalities.
	
	The guarantees provided by Information Density are therefore conditional. Given (i) sufficient historical data to estimate correlations, (ii) approximate stationarity or slow temporal drift in sensing relationships, and (iii) an expressive inference model, high Information Density implies that virtual sensing can achieve low reconstruction or inference error for the corresponding phenomenon. Conversely, low Information Density signals that physical sensing resources provide complementary information that cannot be reliably substituted.
	
	From a deployment perspective, Information Density enables the following trade-offs to be made explicit: (a) reducing physical sensor count versus tolerable inference error, (b) computation and communication overhead versus sensing redundancy, and (c) robustness to sensor failures versus deployment cost. Rather than prescribing a single optimal configuration, the metric supports what-if analysis by ranking sensors or modalities according to their informational contribution under a specified error tolerance. A practitioner workflow diagram is presented in Fig. \ref{fig:workflow}.
	
	Importantly, Information Density does not claim optimality in adversarial, highly non-stationary, or data-scarce environments. In such cases, correlation estimates may be unreliable and virtual sensing performance may degrade. These limits are explicitly discussed in Section \ref{use_case}. Within its intended scope, however, Information Density provides a principled and interpretable mechanism to guide sensor deployment and replacement decisions in large-scale, data-rich sensing infrastructures.

	\textbf{\textit{Relationship Between Information Density and Sensing Performance}}: The measures of Information Density proposed in this work, principal eigenvector alignment and mutual information, do not directly optimize sensing or inference performance. Instead, they characterize structural properties of sensor data that are known to constrain achievable inference accuracy.
	
	In the intra-modality setting, alignment of dominant eigenvectors reflects similarity in the principal variance directions of sensor signals. When such alignment is high, a function approximator trained on a subset of sensors can capture most of the variance necessary to reconstruct the remaining signals, thereby reducing estimation error. Conversely, orthogonal dominant components indicate complementary structure that cannot be inferred without additional physical measurements.
	
	In the cross-modality setting, mutual information provides an upper bound on the predictability of one modality from another, independent of the specific inference model used. While high mutual information does not guarantee low inference error for a particular model, it indicates the existence of a functional relationship that can, in principle, be exploited by sufficiently expressive learners.
	
	These relationships hold most strongly under approximate stationarity and sufficiently rich training data. In non-stationary or highly non-linear environments, Information Density should therefore be interpreted as an empirical indicator rather than a theoretical guarantee of optimal sensing performance.

	\textbf{\textit{Limitations and Failure Modes}}: To be noted that the proposed framework has several important operational scopes and limits. First, Information Density relies on historical data to estimate statistical relationships; rapid non-stationarity or concept drift can invalidate these estimates and degrade virtual sensing performance. Second, mutual information estimates become unreliable under sparse data conditions, limiting applicability in newly deployed or intermittently sampled networks. Third, weak or spurious correlations across modalities reduce the effectiveness of cross-modality inference.
	
	These limits delineate the operational scope of the proposed approach and highlight directions for future work, including adaptive recomputation strategies and online drift detection.

	\textbf{\textit{Computational and Operational Considerations}}: Information Density computation is performed offline or periodically at edge or cloud nodes, rather than on resource-constrained sensors. Eigen-space-based measures scale linearly with sensor count and feature dimension, while mutual information incurs higher computational cost due to pairwise estimation.
	
	In typical urban deployments, Information Density can be recomputed on the order of hours or days, amortizing its cost over long-term sensor operation. This overhead is offset by reduced communication, maintenance, and energy costs resulting from fewer active physical sensors.

	\begin{table}[]\centering
		\caption{Deep Learning Model Baseline Hyperparameters}
		\label{model_param}
		\renewcommand{\arraystretch}{2}
		\begin{tabular}{|c|c|}
			\hline
			\textbf{Parameter}      & \textbf{Value} \\ \hline 
			Training Loss           & MSE            \\ \hline 
			Hidden Layer Activation & ReLU           \\ \hline 
			Learning Rate           & 0.001          \\ \hline 
			Batch Size              & 64             \\ \hline 
			Train-Validation Split        & 80:20          \\ \hline
			Stopping Criteria       & Patience of 500 epochs\\ \hline
		\end{tabular}
	\end{table}

	\section{Use Case: Smart City Sensor Network in Madrid}
	\label{use_case}
	
	\begin{figure}[h]
		\centering
		\includegraphics[height=5.0 cm]{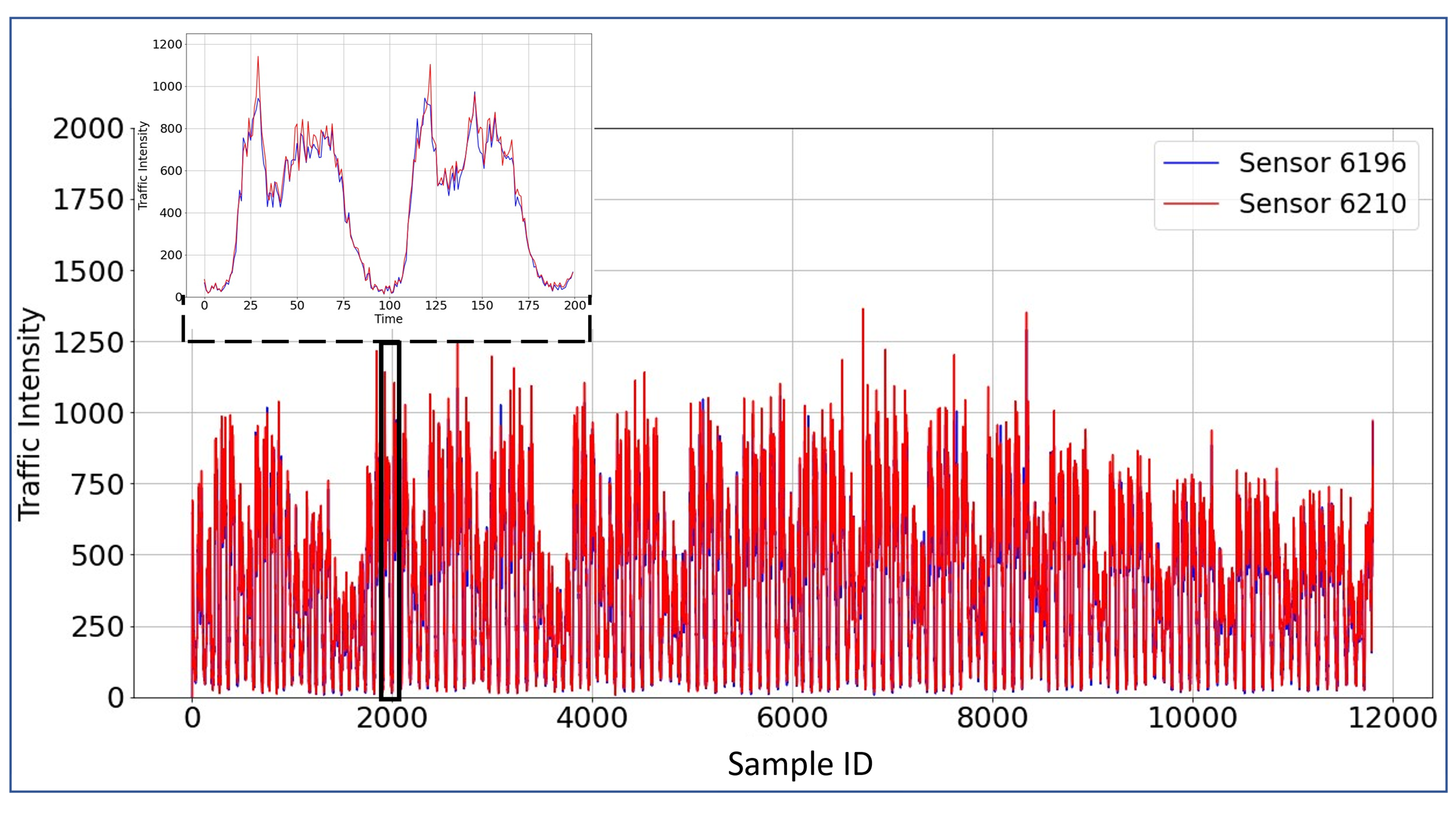}
		\caption{Time-series comparison of traffic readings from two sensors (IDs: 6196 and 6210) with high mutual information and low phase angle. The visual overlap confirms strong redundancy, suggesting one sensor could be virtually inferred from the other.}
		\label{fig:sensor_time_series}
	\end{figure}

	\begin{table*}[!ht]
		\centering
		\footnotesize
		\setlength{\tabcolsep}{4pt}	
		\caption{Performance metrics for virtual sensors under different physical sensor configurations (Scenarios 1–3). Each scenario represents a different number of physical sensors ($|\mathcal{S}_{opt}|$ = 1, 3, 5) used for training, selected based on high information density. Virtual sensors are inferred using a neural network model, and accuracy is assessed using Mean Absolute Error (MAE), Normalized MAE (NMAE), and R² (minimum and maximum values are shown in bold). Results demonstrate that even a single physical sensor can yield low average estimation error ($<3.21\%$).}
		\begin{adjustbox}{width=\textwidth}
			\begin{tabular}{@{}lcccccccccccc@{}}
				\toprule
				\multirow{2}{*}{Virtual Sensors ID} & \multicolumn{4}{c}{Scenario 1} & \multicolumn{4}{c}{Scenario 2} & \multicolumn{4}{c}{Scenario 3} \\
				\cmidrule(lr){2-5} \cmidrule(lr){6-9} \cmidrule(lr){10-13}
				& Physical Sensors ID & MAE & R2 & NMAE & Physical Sensors ID & MAE & R2 & NMAE & Physical Sensors ID & MAE & R2 & NMAE \\
				\midrule
				3695 & 6506 & 20.2262 & 0.8711 & 0.0175 & 6196 & 18.3625 & 0.8900 & 0.0158 & 6196 & 18.0590 & 0.8947 & 0.0162 \\
				3696 & - & 18.6045 & 0.7650 & 0.0161 & 6207 & 17.5504 & 0.7885 & 0.0151 & 6207 & 17.0597 & 0.7953 & 0.0153 \\
				3782 & - & 31.8930 & 0.8197 & 0.0276 & 6506 & 28.9177 & 0.8425 & 0.0249 & 6210 & 28.1554 & 0.8534 & 0.0252 \\
				3783 & - & \textbf{67.5077} & 0.8785 & \textbf{0.0584} & - & \textbf{59.2489} & 0.8989 & \textbf{0.0509} & 6214 & \textbf{57.4459} & 0.9052 & \textbf{0.0515} \\
				6194 & - & 74.1403 & 0.8836 & 0.0641 & - & 49.8684 & 0.9414 & 0.0429 & 6506 & 48.1629 & 0.9467 & 0.0432 \\
				6195 & - & 31.6807 & 0.8059 & 0.0274 & - & 28.4296 & 0.8287 & 0.0244 & - & 27.4203 & 0.8464 & 0.0246 \\
				6196 & - & 57.0584 & 0.8983 & 0.0494 & - & - & - & - & - & - & - & - \\
				6197 & - & 48.3352 & 0.8868 & 0.0418 & - & 39.7293 & 0.9181 & 0.0341 & - & 38.9244 & 0.9227 & 0.0349 \\
				6198 & - & 33.7377 & 0.8507 & 0.0292 & - & 30.1440 & 0.8800 & 0.0259 & - & 29.4285 & 0.8831 & 0.0264 \\
				6199 & - & 25.5660 & 0.8199 & 0.0221 & - & 23.2747 & 0.8441 & 0.0200 & - & 23.2498 & 0.8464 & 0.0208 \\
				6200 & - & 23.7761 & 0.8735 & 0.0206 & - & 21.3544 & 0.8954 & 0.0184 & - & 20.9408 & 0.8997 & 0.0188 \\
				6201 & - & 24.0228 & 0.6796 & 0.0208 & - & 23.0388 & 0.6958 & 0.0198 & - & 23.3100 & 0.6949 & 0.0209 \\
				6202 & - & 21.4668 & 0.8283 & 0.0186 & - & 19.5548 & 0.8498 & 0.0168 & - & 19.6617 & 0.8504 & 0.0176 \\
				6203 & - & 30.9802 & 0.8187 & 0.0268 & - & 28.2692 & 0.8379 & 0.0243 & - & 27.6217 & 0.8483 & 0.0248 \\
				6204 & - & 22.4119 & 0.8490 & 0.0194 & - & 20.4074 & 0.8773 & 0.0175 & - & 19.8447 & 0.8828 & 0.0178 \\
				6205 & - & 26.4881 & 0.8537 & 0.0229 & - & 23.2267 & 0.8849 & 0.0200 & - & 22.8904 & 0.8893 & 0.0205 \\
				6206 & - & 45.5323 & 0.8961 & 0.0394 & - & 20.8375 & \textbf{0.9773} & 0.0179 & - & 22.2160 & \textbf{0.9760} & 0.0199 \\
				6207 & - & 41.5325 & 0.9036 & 0.0359 & - & - & - & - & - & - & - & - \\
				6208 & - & 45.2677 & 0.8977 & 0.0392 & - & 36.8211 & 0.9279 & 0.0316 & - & 32.4175 & 0.9453 & 0.0291 \\
				6209 & - & 45.2657 & 0.8882 & 0.0392 & - & 37.5464 & 0.9158 & 0.0323 & - & 37.6357 & 0.9173 & 0.0337 \\
				6210 & - & 59.3736 & 0.8988 & 0.0514 & - & 22.5204 & 0.9825 & 0.0194 & - & - & - & - \\
				6211 & - & 16.4869 & 0.7549 & 0.0143 & - & 16.0631 & 0.7610 & 0.0138 & - & 15.6807 & 0.7707 & 0.0141 \\
				6212 & - & 39.6106 & 0.8185 & 0.0343 & - & 36.1989 & 0.8355 & 0.0311 & - & 35.3288 & 0.8486 & 0.0317 \\
				6213 & - & 13.3989 & 0.7550 & 0.0116 & - & 12.3854 & 0.7797 & 0.0106 & - & 12.0313 & 0.7896 & 0.0108 \\
				6214 & - & 47.5855 & \textbf{0.9166} & 0.0412 & - & 38.4128 & 0.9436 & 0.0330 & - & - & - & - \\
				6215 & - & 55.6608 & 0.8790 & 0.0482 & - & 39.5701 & 0.9405 & 0.0340 & - & 38.3888 & 0.9473 & 0.0344 \\
				6216 & - & 62.6021 & 0.7667 & 0.0542 & - & 58.7572 & 0.7781 & 0.0505 & - & 53.1117 & 0.8193 & 0.0476 \\
				6217 & - & 50.0275 & 0.8364 & 0.0433 & - & 42.6310 & 0.8779 & 0.0366 & - & 40.0634 & 0.8958 & 0.0359 \\
				6218 & - & 59.4434 & 0.7114 & 0.0514 & - & 53.0078 & 0.7560 & 0.0456 & - & 49.3269 & 0.8060 & 0.0442 \\
				6219 & - & \textbf{7.5969} & \textbf{0.6166} &\textbf{ 0.0066} & - & \textbf{7.3193} & \textbf{0.6462} & \textbf{0.0063} & - & \textbf{7.0307} & \textbf{0.7082} & \textbf{0.0063} \\
				6506 & - & - & - & - & - & - & - & -& - & - & - & - \\
				6922 & - & 32.2178 & 0.8558 & 0.0279 & - & 29.0078 & 0.8748 & 0.0249 & - & 28.4190 & 0.8809 & 0.0255 \\
				6923 & - & 15.7615 & 0.8568 & 0.0136 & - & 14.4797 & 0.8764 & 0.0124 & - & 14.3585 & 0.8783 & 0.0129 \\
				6930 & - & 24.5990 & 0.8097 & 0.0213 & - & 22.6410 & 0.8298 & 0.0195 & - & 22.1709 & 0.8367 & 0.0199 \\
				10400 & - & 60.5447 & 0.7848 & 0.0524 & - & 53.7111 & 0.8229 & 0.0462 & - & 48.2224 & 0.8700 & 0.0432 \\
				10943 & - & 19.2615 & 0.7210 & 0.0167 & - & 18.4379 & 0.7361 & 0.0158 & - & 18.2501 & 0.7375 & 0.0164 \\
				\midrule
				Average Performance & - & 37.1333 & 0.8271 & 0.0321 & - & 30.0523 & 0.8526 & 0.0258 & - & 28.9299 & 0.8576 & 0.0259 \\
				\bottomrule
			\end{tabular}
		\end{adjustbox}
		\label{tab:sensor_performance_part1}
	\end{table*}

	\subsection{Scenario Description}

	To validate the measures in Section~\ref{measures}, we examine a real-world smart city use case in Madrid \footnote{\url{https://datos.madrid.es/portal/site/egob/}}. The city operates a large-scale IoT network that collects spatiotemporal data on traffic, air quality, noise, weather, and energy use. Traffic is monitored using inductive loop detectors, fixed/PTZ cameras, Doppler radar, and LiDAR to compute metrics such as traffic density, flow rate, and average speed, critical for intelligent transportation systems (ITS) \cite{lana2016role}. Environmental data is gathered via meteorological stations equipped with instruments like thermometers, anemometers, and hygrometers. Madrid’s 24 air quality stations monitor pollutants (e.g., $PM_{2.5}, PM_{10}$, $NO_X$, $CO_2, SO_2, O_3$) using specialized analyzers such as chemiluminescence and TEOM devices \cite{nunez2019statistical}. Noise pollution is tracked with precision microphones and frequency analyzers to identify urban acoustic sources \cite{asensio2020changes}.

	From a connectivity standpoint, roadside sensors and IoT-enabled devices relay data wirelessly to centralized hubs. Low-power wide-area network (LPWAN) technologies, such as LoRaWAN and NB-IoT, are commonly utilized for long-range, low-bandwidth data transmission. Meanwhile, real-time communication is often facilitated through Wi-Fi and cellular networks (e.g., GPRS, 4G, 5G), enabling seamless data integration with centralized analytics platforms.	In scenarios where terrestrial infrastructure is limited or absent, satellite links are employed to backhaul environmental data to either local processing units or cloud-based systems. Data aggregation and analysis are handled through integrated IoT platforms like FIWARE and proprietary smart city control centers, with cloud-based infrastructures supporting real-time processing and open APIs for access.
	
	The proposed framework does not explicitly incorporate geometric distance or spatial constraints into the sensor selection process. Instead, spatial effects are implicitly reflected through the statistical structure of the collected data. Sensors deployed in close proximity or subject to similar environmental influences tend to exhibit stronger correlations and higher Information Density, whereas sensors that are spatially distant or exposed to heterogeneous conditions typically contribute complementary information. The experiments in this paper are conducted on real-world sensing datasets with fixed spatial deployments. Consequently, the reported results evaluate the effectiveness of Information Density in guiding sensor selection and virtual sensing under an existing deployment, rather than optimizing spatial placement.

	\textbf{\textit{Data Collection, Processing, and Reproducibility Considerations}}: As mentioned earlier, the experimental evaluation in this paper is conducted using an open urban sensing dataset released by the City of Madrid through its public open data portal. The dataset is collected and maintained by the city’s municipal sensing infrastructure and made available as periodically updated time-series records. Our study uses a static snapshot of the dataset corresponding to a fixed observation window, rather than live data streams.
	
	Data acquisition and preprocessing were performed offline using standard data analytics tools. Raw sensor readings were temporally aligned, normalized, and filtered to remove incomplete records prior to analysis. No proprietary software or specialized hardware is assumed; all processing steps can be reproduced using commonly available data science frameworks.
	
	The dataset follows the FAIR data principles. It is findable and accessible through a persistent public repository provided by the City of Madrid, interoperable through standardized data formats and metadata descriptions, and reusable under an open license that permits research use. These properties support reproducibility of the reported results and enable independent validation of the proposed framework.
	
	We emphasize that the proposed Information Density framework is agnostic to the specific data collection infrastructure and can be applied to other open or proprietary sensing datasets that provide sufficient historical data for statistical analysis.

	\begin{figure*}[h]
		\centering
		\includegraphics[width=\textwidth]{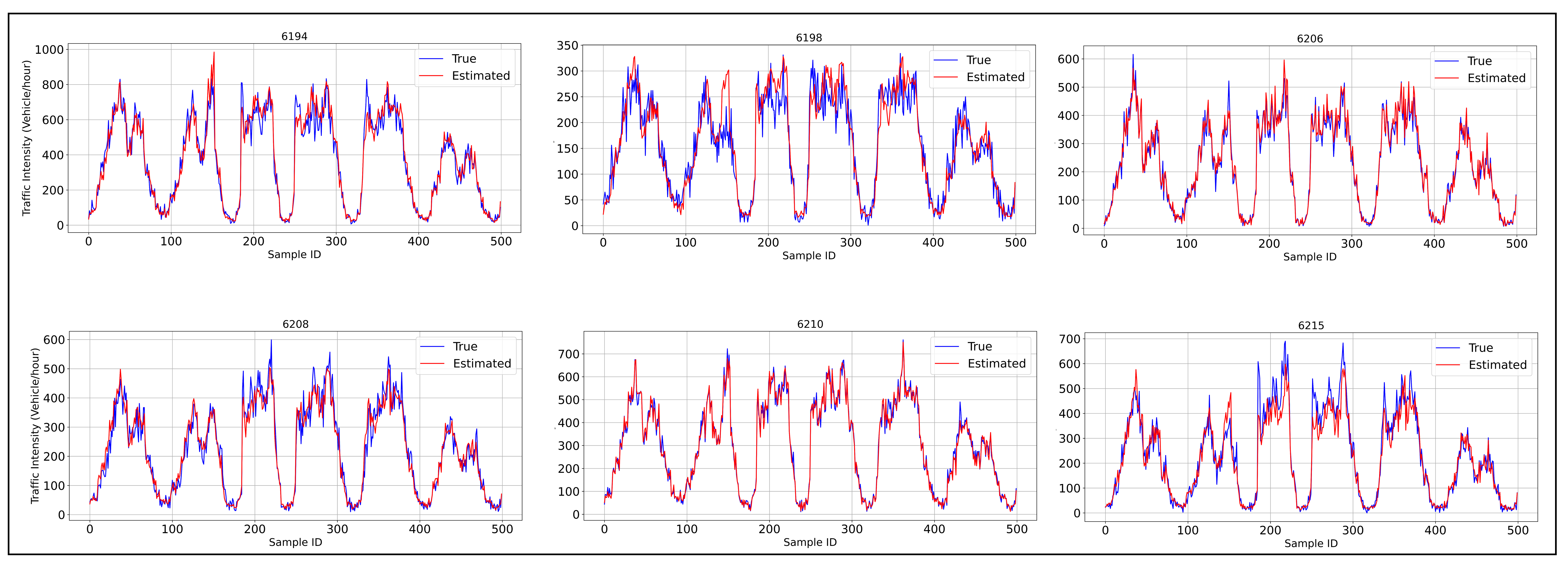}
		\caption{Predicted vs. actual sensor readings in the Intra-modality Virtual Sensing (ImVS) experiment. The plots demonstrate that the deep learning model accurately reconstructs virtual sensor signals using a limited number of physical sensors with high information density.}
		\label{fig:spatial_dmf_time_series}
	\end{figure*}

	\begin{table*}[!ht]
		\centering
		\scriptsize
		\renewcommand{\arraystretch}{0.7}
		\setlength{\tabcolsep}{3.5pt}
		
		\caption{Performance metrics for virtual sensors under extended configurations (Scenarios 4–5).}
		\begin{adjustbox}{width=14 cm}
			\begin{tabular}{@{}lcccccccc@{}}
				\toprule
				\multirow{2}{*}{Virtual Sensors ID} & \multicolumn{4}{c}{Scenario 4 (8 sensors)} & \multicolumn{4}{c}{Scenario 5 (10 sensors)} \\
				\cmidrule(lr){2-5} \cmidrule(lr){6-9}
				& Physical Sensors ID & MAE & R2 & NMAE & Physical Sensors ID & MAE & R2 & NMAE \\
				\midrule
				3695 & 6196 & 18.2850 & 0.8922 & 0.0162 & 6196 & 18.0270 & 0.8948 & 0.0158 \\
				3696 & 6197 & 16.4432 & 0.8178 & 0.0146 & 6197 & 16.3065 & 0.8239 & 0.0143 \\
				3782 & 6206 & 25.2405 & 0.8841 & 0.0224 & 6206 & 23.2097 & 0.9024 & 0.0203 \\
				3783 & 6207 & 35.6950 & 0.9655 & 0.0317 & 6207 & - & - & - \\
				6194 & 6208 & 41.9510 & 0.9645 & 0.0372 & 6208 & \textbf{41.2892} & \textbf{0.9662} & \textbf{0.0362} \\
				6195 & 6210 & 24.1813 & 0.8856 & 0.0214 & 6210 & 23.4502 & 0.8925 & 0.0206 \\
				6196 & 6214 & - & - & - & 6214 & - & - & - \\
				6197 & 6506 & - & - & - & 6506 & - & - & - \\
				6198 & - & 29.0999 & 0.8860 & 0.0258 & 6209 & 29.3407 & 0.8860 & 0.0257 \\
				6199 & - & 22.6370 & 0.8526 & 0.0201 & 3783 & 22.3167 & 0.8561 & 0.0196 \\
				6200 & - & 20.8853 & 0.8991 & 0.0185 & - & 21.0945 & 0.8987 & 0.0185 \\
				6201 & - & 21.0577 & 0.7324 & 0.0187 & - & 21.0122 & 0.7347 & 0.0184 \\
				6202 & - & 18.1953 & 0.8719 & 0.0161 & - & 18.1633 & 0.8742 & 0.0159 \\
				6203 & - & 24.1445 & 0.8846 & 0.0214 & - & 23.6170 & 0.8900 & 0.0207 \\
				6204 & - & 19.5266 & 0.8874 & 0.0173 & - & 19.6756 & 0.8866 & 0.0172 \\
				6205 & - & 22.7487 & 0.8893 & 0.0202 & - & 22.8186 & 0.8906 & 0.0200 \\
				6206 & - & - & - & - & - & - & - & - \\
				6207 & - & - & - & - & - & - & - & - \\
				6208 & - & - & - & - & - & - & - & - \\
				6209 & - & 19.5365 & \textbf{0.9772} & 0.0173 & - & - & - & - \\
				6210 & - & - & - & - & - & - & - & - \\
				6211 & - & 15.6184 & 0.7721 & 0.0139 & - & 16.2318 & 0.7577 & 0.0142 \\
				6212 & - & 28.7409 & 0.9025 & 0.0255 & - & 28.0762 & 0.9089 & 0.0246 \\
				6213 & - & 12.3998 & 0.7909 & 0.0110 & - & 11.8537 & 0.8045 & 0.0104 \\
				6214 & - & - & - & - & - & - & - & - \\
				6215 & - & 36.1353 & 0.9524 & 0.0320 & - & 35.0045 & 0.9558 & 0.0307 \\
				6216 & - & \textbf{42.8128} & 0.8897 & \textbf{0.0380} & - & 40.3368 & 0.9052 & 0.0354 \\
				6217 & - & 35.5553 & 0.9247 & 0.0315 & - & 32.9761 & 0.9352 & 0.0289 \\
				6218 & - & 38.8889 & 0.8937 & 0.0345 & - & 36.5927 & 0.9092 & 0.0321 \\
				6219 & - & \textbf{7.1913} & \textbf{0.7041} & \textbf{0.0064} & - & \textbf{6.9706} & \textbf{0.7176} & \textbf{0.0061} \\
				6506 & - & - & - & - & - & - & - & - \\
				6922 & - & 24.8756 & 0.9088 & 0.0221 & - & 24.2361 & 0.9145 & 0.0212 \\
				6923 & - & 14.2668 & 0.8764 & 0.0127 & - & 14.3508 & 0.8763 & 0.0126 \\
				6930 & - & 21.8101 & 0.8399 & 0.0193 & - & 21.6041 & 0.8429 & 0.0189 \\
				10400 & - & 41.2831 & 0.9135 & 0.0366 & - & 37.6375 & 0.9306 & 0.0330 \\
				10943 & - & 17.6647 & 0.7438 & 0.0157 & - & 16.9470 & 0.7768 & 0.0149 \\
				\midrule
				Average Performance & - & 24.8882 & 0.8715 & 0.0221 & - & 23.9669 & 0.8705 & 0.0210 \\
				\bottomrule
			\end{tabular}
		\end{adjustbox}
		\label{tab:sensor_performance_part2}
	\end{table*}

	\begin{figure*}[!ht]
		\centering
		\includegraphics[width=\textwidth]{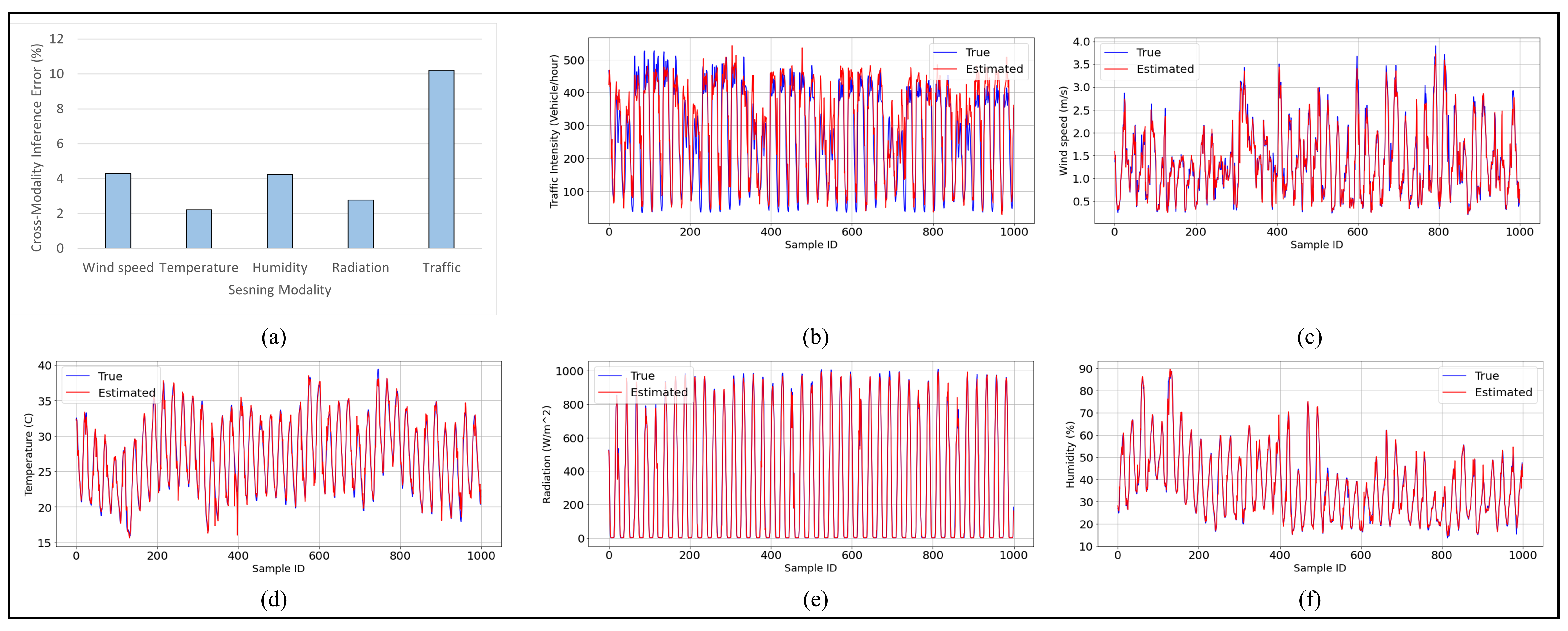}
		\caption{Estimated signal performance in cross-modality inference: (a) Normalized Mean Absolute Error (NMAE) for estimating various environmental modalities (e.g., temperature, humidity, radiation) from pollutant concentration data using cross-modality inference. Lower values indicate accurate virtual sensing, and (b)-(d) Examples of actual vs. estimated signals for different environmental modalities using cross-modality virtual sensing. The close match in trends and amplitudes illustrates the model’s capability to infer multiple phenomena from pollutant data alone.}
		\label{fig:cmi_time_series}
	\end{figure*}

	\subsection{Intra-modality Virtual Sensing}
	\label{imvr_res}
	
	In this study, we first extend the concept of ``Data-driven Modality Fusion (DMF)" to a more generalized setting of Virtual Sensing that includes both Intra-modality Virtual Sensing (ImVS) and Cross-modality Inference (CmI), as explained earlier in section \ref{measures}. We first demonstrate the concept of Intra-modality Virtual Sensing in this setting and analyze the measures of Information Density in this context. For this analysis, we first consider the traffic sensor network in the district 19, with 45 sensing units, of Madrid collecting traffic intensity readings at an interval of 15 minutes for the entire year of 2023.
	
	%

	The 2-D scalar fields representing information density in this region are depicted in Figs. \ref{fig:phase_2d} and \ref{fig:mi_2d} that use the definitions of Eigen space phase angles and mutual information, defined in Eqn. \ref{angle_comp} and \ref{mi__eqn} respectively. The prima-facie observation from this is that both these figures complement each other by representing the structural similarity in the sensor readings for that specific area, in the form of information density. Although these two metrics provide the similar representation regarding the sensor measurements, however, they are inverse in characteristics. In other words, information density in the Eigen space indicate divergence measures, while that in the mutual information space demonstrate a similarity metric. Another noteworthy observation from these figures is the high information density existing between sensors 6196 and 6210, represented by their very low phase angle (blue in \ref{fig:phase_2d}) and very high mutual information (red in \ref{fig:mi_2d}) values. In order to validate this, we plotted the readings from these two sensors in the same frame, as shown in Fig. \ref{fig:sensor_time_series}. The plot also demonstrates the strong similarity in the sensor readings that exists between them. From an application or implementation standpoint, this holds an important significance, stating redundancy in deploying both these sensors. In other words, given such strong information density between these two, one can be replaced with the others, with a given tolerable margin of error.

	\begin{figure}[h]
		\centering
		\includegraphics[height=5.0 cm]{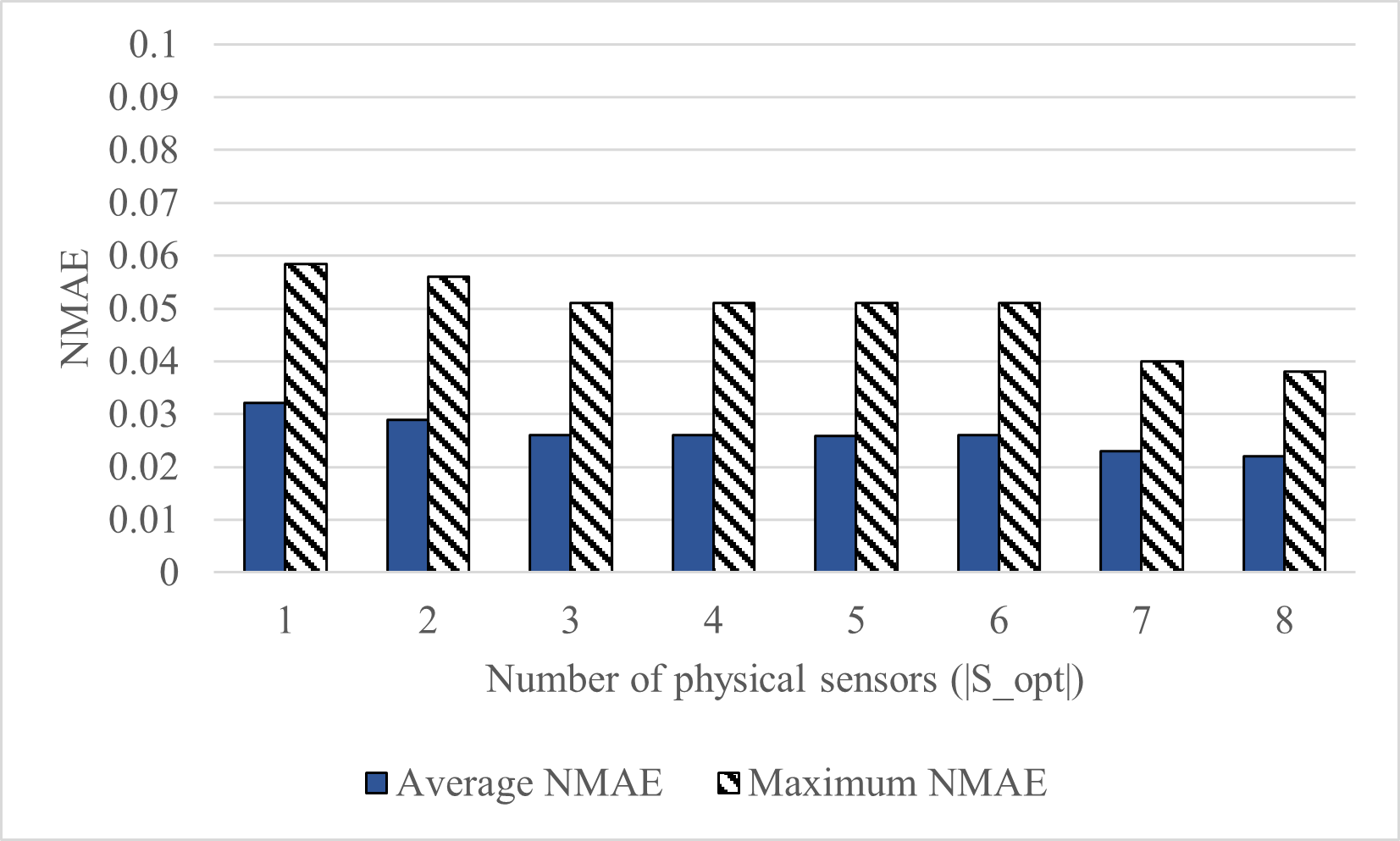}
		\caption{Estimation errors for virtual sensing for varying number of physical sensors}
		\label{fig:s_opt}
	\end{figure}

	Having said that, the next step is to demonstrate how to leverage this information density scalar field to implement Intra-modality Virtual Sensing (ImVS). As explained earlier, this paradigm of ImVS is an extension of Data-driven Modality Fusion for a generalized spatial sensing context. This can be accomplished by first choosing $|\mathcal{S}_{opt}|$ sensors with high information density measures (Eqns. \ref{phase_min_eqn} and \ref{mi_max_eqn}) and then modeling their interdependence using a function approximation framework. In this setup, we have used a deep learning architecture for modeling this inter-dependency among these sensors with high information density. The details about the deep learning model used and the related hyperparameters are provided in Table \ref{model_param}.

	The deep learning model uses feed-forward Neural Network (NN) with 7 hidden layers. The structure of the NN used in this set of experiments is given by $|\mathcal{S}_{opt}|\times 32\times 64\times 128\times 256\times 512\times 256\times 64\times (|\tilde{\mathcal{S}}|-|\mathcal{S}_{opt}|)$. The models are trained by minimizing the  Mean Square Error (MSE) losses defined in Eqns. \ref{mae_itr}. The MSE loss for a model estimating the readings of sensor $s_m$ (for a training batch of size $N$), that needs to be minimized, for this ITR approach is defined by Eqn. \ref{mae_itr}.
	
	\begin{equation}
		\mathcal{L}(s_m; \theta_k) = \sum_{i=1}^{N}{(f_m(\mathbf{x}_i; \theta_k)-y_{s_m})^2};  1\leq m \leq |\mathcal{S}_{opt}|
		\label{mae_itr}
	\end{equation}
	Here, $\theta_k$ is the set of model parameters at instance $k$ and $f_m(\mathbf{x}; \theta_k)$ is the function approximated for the $m^{th}$ sensor for the input predictor vector $\mathbf{x}$.
	
	The loss functions are minimized using the Adaptive Moment Estimation Algorithm that computes individual adaptive learning rates for different parameters from estimates of first and second moments of the gradients \cite{kingma2014adam}. The model parameter update rule at instance $k$ is given by Eqn. \ref{adam}, where $g_k, v_k$ and $m_k$ represent the loss gradient, first and second order moment estimates respectively.
	
	\begin{subequations}
		\label{adam}
		\begin{align}
			g_k = \nabla \mathcal{L}(\theta_{k-1})\\
			m_k = \beta_1 m_{k-1}+(1-\beta_1)g_k\\
			v_k = \beta_2 v_{k-1}+(1-\beta_2)(g_k)^2\\
			\theta_{k}=\theta_{k-1}-\alpha_k \times \frac{m_k}{(1-\beta_1^k)\times (\sqrt{\frac{v_k}{1-\beta_2^k}}+\epsilon)}
		\end{align}
	\end{subequations}

	We perform experiments with different size of the set of physical sensors ($\mathcal{S}_{opt}$). Tables \ref{tab:sensor_performance_part1} and \ref{tab:sensor_performance_part2} summarize performance under five scenarios where the number of physical sensors used ($|\mathcal{S}_{opt}|$) varies from 1 to 10. Each scenario tests the model’s ability to estimate readings of virtual sensors across different configurations, progressively increasing the size of the physical sensor set. As shown in the tables, we consider five different cases with $|\mathcal{S}_{opt}|=1, 3, 5, 8, 10$ number of physical sensors and estimate the readings of the other sensors in that area. Performance is evaluated using Mean Absolute Error (MAE), its normalized value (NMAE) and R2 score. The noteworthy observation here is that even using only single physical sensor (that is, $\mathcal{S}_{opt}=1$, chosen following Eqns. \ref{phase_min_eqn} and \ref{mi_max_eqn}), the mean error is $3.21 \%$ and the error is bounded by  $6.41 \%$, demonstrating the effectiveness and feasibility of the approach. Moreover, with the increase in the size of the set  $\mathcal{S}_{opt}$, this error reduces, although it saturates beyond a point. Some comparative samples of the estimated and true signals from the sensors are presented in Fig. \ref{fig:spatial_dmf_time_series}, which show the ability of Intra-modality Virtual Sensing to estimate the measurements of the virtual sensors with high accuracy. As seen in the figure, most virtual sensor estimates closely match the true values, demonstrating the high fidelity of the proposed virtual sensing approach. While some discrepancies exist, particularly during peak fluctuations, the overall signal trend and amplitude are preserved, validating the model’s effectiveness in high-density information regions. This is also demonstrated in the estimation error variation plot for varying number of physical sensors, as shown in Fig. \ref{fig:s_opt}.

	\begin{figure}[h]
		\centering
		\includegraphics[height=6.0 cm]{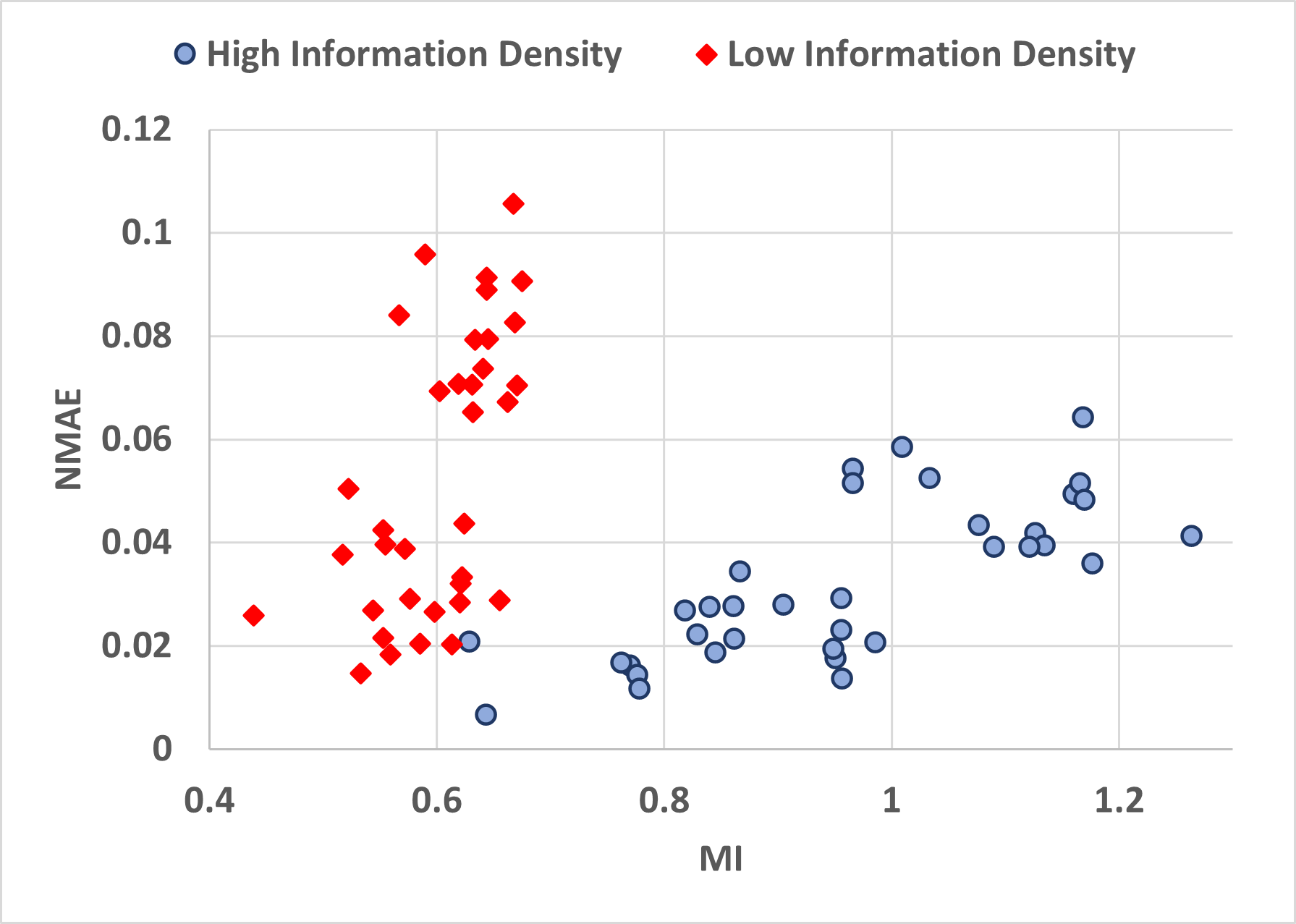}
		\caption{Variation of estimation error across different sensor selections in ImVS. Sensors selected based on high information density yield significantly lower prediction errors, confirming the metric’s utility in sensor selection and system performance.}
		\label{fig:clusters}
	\end{figure}

	Fig. \ref{fig:clusters} shows the relationship between information density and the performance of Intra-modality Virtual Sensing (ImVS) for two scenarios: first, choosing the physical sensors based on high information density following the scalar field of Fig. \ref{fig:comparison_2d} and second, choosing sensors with low information density. It is clearly evident from the figure that selecting physical sensors in ImVS with high information density outperforms the reverse scenario, emphasizing the need for Information Density measures in this context. The figure also provides another insight in regards to the high correlation existing between information density and model performance, that is, selecting high information density sensors for training typically leads to low estimation error.

	\begin{figure}[h]
		\centering
		\includegraphics[height=5.0 cm]{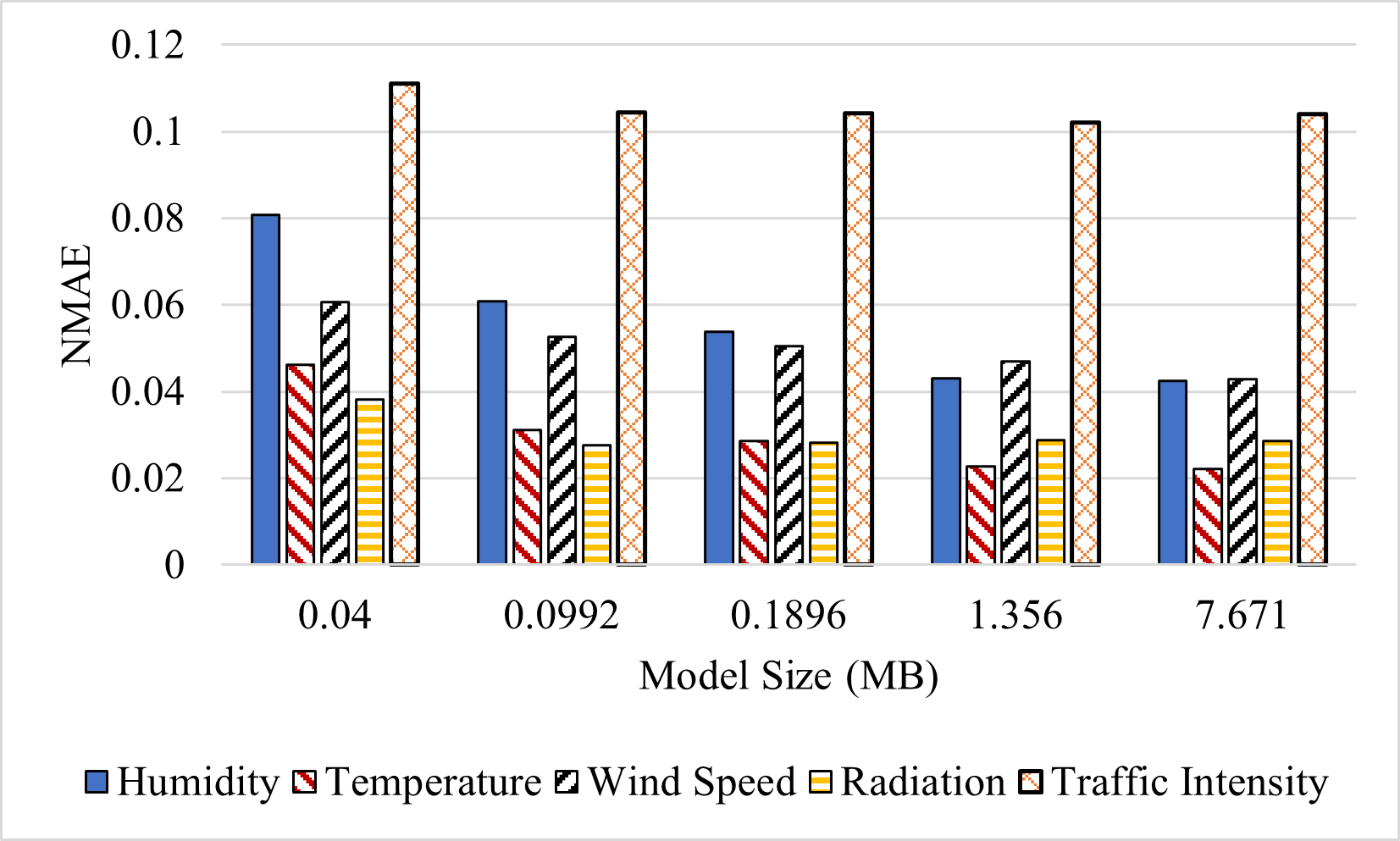}
		\caption{Virtual Sensing performance variation with model architecture}
		\label{fig:model_size}
	\end{figure}

	\begin{table}[h!]
		\centering
		\caption{Information Density measures between pollutant concentration data and other modalities for cross-modality inference. Mutual Information and Phase in Eigen Space (converted to cosine-based Similarity Score) are computed between pollutant data and five other modalities: wind speed, temperature, humidity, radiation, and traffic. A random signal is also included as a baseline. Higher similarity scores and mutual information indicate greater potential for accurately inferring one modality from another, validating the selection of pollutants as input features for virtual sensing}
		\begin{tabular}{l c c c}
			\toprule
			\textbf{Parameter} & \makecell{\textbf{Mutual} \\ \textbf{Information}} & \makecell{\textbf{Phase in} \\ \textbf{Eigen Space (degree)}} & \makecell{\textbf{Similarity} \\ \textbf{Score}} \\
			\midrule
			Wind speed  & 0.223083698 & 119.81  & 0.496205 \\
			Temperature & 0.147406851 & 108.22  & 0.311757 \\
			Humidity    & 0.098116551 & 102.60  & 0.217257 \\
			Radiation   & 0.100822997 & 102.71  & 0.219130 \\
			Traffic     & 0.104425675 & 78.51   & 0.199878 \\
			Random      & 0.001387143 & 89.625  & 0.007338 \\
			\bottomrule
		\end{tabular}
		\label{tab:cmi}
	\end{table}

	\begin{table*}[!ht]
		\centering
		\caption{Qualitative comparison of sensor selection methods across accuracy-relevant and system-level criteria.}
		\label{tab:baseline_comparison}
		\renewcommand{\arraystretch}{1.2}
		\setlength{\tabcolsep}{6pt}
		\begin{tabular}{lcccc}
			\hline
			\textbf{Criterion} 
			& \textbf{Random} 
			& \textbf{Correlation-based \cite{ding2024deep}} 
			& \textbf{Variance-based  \cite{wang2024dynamic}} 
			& \textbf{Information Density (Proposed)} \\
			\hline
			Captures dominant variance structure 
			& \texttimes 
			& \texttimes 
			& \checkmark 
			& \checkmark \\
			
			Accounts for redundancy among sensors 
			& \texttimes 
			& Partial 
			& Partial 
			& \checkmark \\
			
			Robust to noisy or spurious correlations 
			& \texttimes 
			& \texttimes 
			& \checkmark 
			& \checkmark \\
			
			Handles non-linear dependencies 
			& \texttimes 
			& \texttimes 
			& \texttimes 
			& \checkmark$^{\dagger}$ \\
			
			Applicable to cross-modality selection 
			& \texttimes 
			& \texttimes 
			& \texttimes 
			& \checkmark \\
			
			Computational cost 
			& Very low 
			& Low 
			& Moderate 
			& Moderate \\
			
			Data requirement 
			& Low 
			& Low 
			& Moderate 
			& Moderate \\
			\hline
		\end{tabular}
		\begin{flushleft}
			\footnotesize
			$^{\dagger}$When instantiated using mutual information; eigen-space measures capture linear structure.
		\end{flushleft}
	\end{table*}

	\subsection{Cross-modality Inference}
	
	To analyze the importance of information density measures on cross-modality inference, we consider all the sensing stations of Madrid. We take into account six different sensing modalities for this purpose: traffic intensity, pollutant concentrations ($SO_2, CO, NO,PM_{2.5},PM_{10},NO_X, O_3$, $C_6H_6$, $C_6H_5CH_3$, $C_8H_{10}$), wind speed, temperature, radiation and humidity. We consider the problem of estimating all other modalities from pollutant concentration measurements. Hence, we first compute the information density measures between the pollutant concentration readings and all other parameters, as shown in Table \ref{tab:cmi}. It is worth mentioning that we have adapted the information density measure based on Eigen vector angle $\omega_{m.n}$ between two modalities $m, n$ by defining a similarity score $\tau_{m,n}$, given by Eqn. \ref{cmi_eqn}.
	
	\begin{equation}
		\tau_{m,n}=cos(\omega_{m.n}); 0\leq\tau_{m,n}\leq 1
		\label{cmi_eqn}
	\end{equation}
	
	Note that in Table \ref{tab:cmi}, we have also computed the information density for a randomly generated sensor reading (following a uniform distribution) to strengthen our claim. It can be observed that the phase between the principal Eigen vectors of the random variable and the pollutant concentrations are orthogonal. In other words, there isn't much similarity in information between these two modalities, which is logical. As a result, the similarity score and the information density measure based on mutual information (Eqn. \ref{mi__eqn}) are very low for this scenario; which does not hold for the other modalities, indicating there is a possibility of estimating the rest of the modalities from the pollutant concentration readings.
	
	To validate the claims, we consider the same experimental setup used for ImVS discussed earlier in section \ref{imvr_res}. Here, the deep learning model used for function approximation is also a feed-forward Neural Network (NN) with 8 hidden layers. The structure of the NN used in this set of experiments to estimate the modalities is given by $|\mathcal{M}_{poll}|\times 20\times 50\times 100\times 300\times 300\times 100\times 50\times 20\times 1$. Here $\mathcal{M}_{poll}$ is the set of modalities used for pollutant concentrations. The training of the models follows the same procedure used in ImVS, that is, minimizing the  Mean Square Error (MSE) losses defined in Eqns. \ref{mae_itr}. Fig. \ref{fig:model_size} summarizes the Virtual Sensing performance across different sensing modalities as a function of model complexity. Overall, the estimation error decreases as the model size increases, up to a point beyond which the gains diminish. This trend arises because larger models offer greater degrees of freedom to approximate the nonlinear relationship between each sensing modality and the pollutant concentration. Note that in data-scarce settings, more complex models are prone to overfitting due to the curse of dimensionality, which could be mitigated by incorporating multiple years of training data for sensor reading estimation.

	Fig. \ref{fig:cmi_time_series} demonstrates the performance of the framework in terms of estimation error. defined by MAE. The bar plot in Fig. \ref{fig:cmi_time_series}a shows that the inference error in this set of experiments is bounded by 11\%. The time-series plots in Fig. \ref{fig:cmi_time_series}b-d demonstrate a comparison between the true measurements and the estimated signals using CmI. The plots show strong correlation between the true and estimated readings, thus justifying the significance of Information Density measures used for the choice of sensing modalities used.

	\subsection{Comparison with Baselines}
	
	In addition to random sensor selection, we evaluate two standard baselines: (i) variance-based selection \cite{wang2024dynamic}, where sensors are ranked by signal energy or PCA loading magnitude, and (ii) correlation-based selection \cite{ding2024deep}, where sensors are ranked according to their average absolute Pearson correlation with other sensors, favoring less correlated sensors. These baselines represent commonly used heuristics for sensor reduction in practice.
	
	While correlation-based selection captures pairwise linear relationships, it does not distinguish between dominant and minor variance components and is sensitive to noise and scaling. In contrast, Information Density operates in eigen-space, enabling it to identify redundancy in dominant variance structure and to generalize naturally to cross-modality scenarios. As a result, Information Density provides a more robust and principled basis for sensor selection, particularly in dense sensing environments. A detailed comparison is presented in Table \ref{tab:baseline_comparison}.
	
	\section{Conclusions}
	\label{conc}

	This paper presented Information Density as a novel quantitative measure to inform and structure sensor deployment and replacement decisions, while enabling AI-driven virtual sensing in large-scale IoT networks. By incorporating spatial, temporal and inter-modal correlations into the sensing model, we demonstrated that virtual sensors can effectively replace or supplement physical ones, significantly reducing redundancy and computational overhead. The proposed complementary measures, viz. Phase in Eigen Space and Mutual Information, provided robust criteria for selecting optimal sensor configurations, validated through real-world experiments in Madrid's smart city environment. Key findings include the ability to achieve low estimation errors (e.g., $<6.41 \%$ NMAE with minimal physical sensors and a mean NMAE of $3.21 \%$) and the strong correlation between information density and sensing performance. These insights pave the way for scalable, adaptive, and energy-efficient sensing infrastructures, with applications ranging from smart cities to industrial automation.

	\ifCLASSOPTIONcaptionsoff
	\newpage
	\fi

	\bibliographystyle{IEEEtran}
	\bibliography{bibtex/Reference}

	\begin{IEEEbiography}
		{Hrishikesh Dutta}
		is a Postdoctoral Research Scientist in the Data Intelligence and Communication Engineering Laboratory of Institut Polytechnique de Paris. He received his PhD from Michigan State University in 2024. His research interests include Internet-of-Things, Computer Networks, Machine Learning and Distributed Systems.
	\end{IEEEbiography}
	
	\begin{IEEEbiography}
		{Prof. Roberto Minerva}
		received the M.S. degree (summa cum laude) in computer science from the University of Bari Aldo Moro, Bari, Italy, in 1987, and the Ph.D. degree in computer science
		and telecommunications from Pierre and Marie Curie University–Sorbonne University, Paris, France, in 2013.	From 1987 to 1996, he was a Researcher in the area of service architectures and network intelligence with the Telecom Italia Research Center, Turin, Italy. In the following years, he was responsible for several research groups related to network intelligence and evolution to next generation networks. From 2013 to 2016, he was appointed to the Strategic Initiatives of TIM. Since 2016, he has been the Technical Project Leader of the SoftFIRE, a European Project devoted to the experimentation of network function virtualization (NFV), software defined networking (SDN), and edge computing. Since 2018, he has been an Associate Professor in softwarization with the Service Architecture Laboratory, Wireless Networks and Multimedia Services Department, Institut Mines Telecom, Telecom Sud Paris, Évry, France, a part of the Institute Polytechnique de Paris, Paris. He has authored the book Networks and New Services: A Complete Story. He is also	an author of more than 50 papers in journals and international conferences. Dr. Minerva is a member of the Scientific Committee of the Fondazione Bruno Kessler, Trento, Italy. He has been the Chairperson of the IEEE Internet of Things Initiative from 2014 to 2016			
	\end{IEEEbiography}

	\begin{IEEEbiography}
		{Reza Farahbakhsh}
		received the B.S. degree from Qazvin Azad University in 2006, the M.S. degree in computer engineering from the University of Isfahan, Iran, 2008, and the Ph.D. degree from the CNRS Laboratory UMR5157, Institut-Mines Telecom, Telecom SudParis, jointly with UPMC, Paris VI, in 2015. He is currently a Post-Doctoral Researcher with the Institut-Mines Telecom, Telecom SudParis. His research interests are online social networks, IoT, large scale measurement, and user behavior analysis.
	\end{IEEEbiography}
	
	\begin{IEEEbiography}
		{Prof. Noel Crespi}
		holds Masters degrees from the Universities of Paris-Saclay (formerly Orsay) and Kent (UK), a diplome d’ingénieur from Telecom Paris, and a Ph.D and an Habilitation from Sorbonne University. From 1993 he worked at CLIP, Bouygues Telecom and then at Orange Labs in 1995. He took leading roles in the creation of new services with the successful conception and launch of Orange prepaid service (15M+ subscribers), and in standardization (from the rapporteurship of IN standard to the coordination of all mobile standards activities for Orange). In 1999, he joined Nortel Networks as telephony program manager, architecting core network products for the EMEA region. He joined Institut Mines-Telecom, Telecom SudParis in 2002 and is currently Professor and Program Director at Institut Polytechnique de Paris, leading the Data Intelligence and Communication Engineering Lab. He coordinates the standardization activities for Institut Mines-Telecom at ITU-T and ETSI. He was an adjunct professor at KAIST (South Korea), a guest researcher at the University of Goettingen (Germany) and an affiliate professor at Concordia University (Canada). He is the scientific director of AI2, a French-Korean laboratory. As a Principal Investigator or Co-Investigator, he has secured research grants added up to 8M+ euros. His current research interests are in Edge Intelligence, IoT, Digital Twin, Artificial Intelligence and NLP.
		http://noelcrespi.wp.tem-tsp.eu/			
	\end{IEEEbiography}

	\vfill

\end{document}